\documentclass[a4paper,14pt]{article}
\usepackage{bbm}
\usepackage{amsthm}
\usepackage[english]{babel}
\usepackage{stmaryrd}
\usepackage{hyperref}
\usepackage{a4wide}
\newcommand{\norm}[1]{\left\|{#1}\right\|}
\newcommand{\N}{\mathbbmss{N}}

\newcommand{\C}{\mathbbmss{C}}
\newcommand{\set}[1]{\left\{{#1}\right\}}

\newtheorem{proposition}{Proposition}
\newtheorem{lemma}{Lemma}
\newtheorem{corollary}{Corollary}
\newtheorem{definition}{Definition}
\newtheorem{theorem}{Theorem}
\newtheorem{example}{Example}

\begin{document}
\title{Entanglement of Indistinguishable Particles and its Quantification}
\author{Florian Sokoli, Burkhard K\"ummerer\\Fachbereich Mathematik, Technische Universit\"at Darmstadt}
\maketitle
\begin{abstract}
We introduce geometric measures of entanglement for indistinguishable particles, which apply to mixed states, multipartite systems, and arbitrary dimensions. They are based on generalized (i.e., not necessarily finite) norms on the set of quantum states and lead to the first necessary and sufficient computational separability criterion in this general setting. The coherent approach developed in the paper allows us to compare, in particular, entanglement for fermionic and distinguishable particles: The entanglement measure for fermionic particles coincides with the corresponding entanglement measure for distinguishable particles up to a factor of $k!$ where $k$ is the number of particles involved. By this result the amount of entanglement emerging from fermi statistics alone is clearly separated from the overall amount of entanglement. Finally, our techniques are applied to entanglement related to Schmidt and Slater numbers.	
\end{abstract}


\maketitle
\section{Introduction}
Entanglement of indistinguishable particles has gained an increasing amount of interest during the last decade. Due to the different nature of fermionic and bosonic quantum systems in comparison to distinguishable particles the conventional approach to entanglement needs to be modified considerably. In particular, entanglement of fermionic particles has been be considered as being ''unphysical'' and ''useless'' for some time. Meanwhile, the conceptional problem of defining entanglement of indistinguishable particles in an appropriate way has widely been solved \cite{Benatti,Eckert,Ghi02,Grabowski,Li,Schliemann,Tichy}.\\
\indent However, the possibilities for detecting and quantifying this kind of entanglement are still rather limited and many questions are open.\\
\indent Schliemann et al. presented the first characterization of entanglement for quantum states of indistinguishable particles. They studied entanglement of mixed states of bipartite quantum systems using a \emph{correlation measure}, also called \emph{Schliemann concurrence} \cite{Eckert,Schliemann}. A more general and widely used approach is to consider the reduced one-particle-state of a given multipartite state of indistinguishable particles. In this context several entanglement measures have been proposed such as the \emph{purity} or the \emph{von Neumann entropy} of the reduced state \cite{Balachandran, Ghi04, Ghi05, Iemini2, Paskauskas, Plastino, Zander}. It can shown that these techniques are sufficient in order to characterize entanglement of pure fermionic states \cite{Plastino}. Moreover, measures based on a \emph{shifted negativity} \cite{Iemini2} and \emph{R\'enyi entropies} \cite{Zander} have been used in this context as well. L\'evay and Vrana \cite{Levay} have introduced an algebraic measure for tripartite fermionic entanglement as well as a separability criterion for pure multipartite states based on methods from algebraic geometry. Finally, Iemini et al. \cite{Iemini1} have considered \emph{entanglement witnesses} in order to obtain an entanglement measure called \emph{generalized robustness}. However, up to now there is no necessary and sufficient separability criterion for indistinguishable particles which works for arbitrary mixed and multipartite states in arbitrary dimensions. The present paper contributes to this problem by introducing measures for fermionic and bosonic entanglement in this general setting which lead to necessary and sufficient separability criteria.\\
\\
Comparing the various definitions of entanglement for distinguishable, bosonic and fermionic particles one realizes a great deal of common mathematical structure which indicates the possibility of a more abstract and unified treatment. In a groundbreaking work \cite{Arveson} Arveson provided such a theory of generalized entanglement which allows a coherent approach to the detection and quantification of all different kinds of quantum entanglement at once. The general idea is to use \emph{Minkowski functionals} with respect to convex balanced sets as entanglement measures. It turns out that they can be interpreted as \emph{generalized norms} in the sense that all properties of a norm are satisfied except that they might be infinite for infinite dimensional Hilbert spaces. Although these concepts are already very general they still need to be adapted in order to control entanglement of quantum states restricted to linear subspaces. We therefore develop a slight generalization of Arveson's concepts which, in particular, leads to a simper and easier-to-handle theory. 
\\
\indent A widely used treatment of entanglement of fermionic particles is to consider them as being distinguishable and to study their entanglement in the ''classical'' way as for distinguishable particles. Our treatment of entanglement allows us to provide a mathematical justification of this approach. We show that the entanglement measure for fermions coincides with the corresponding entanglement measure for distinguishable particles up to a state-independent factor of $k!$ where $k$ equals the number of particles involved. Moreover, the amount of ''trivial'' entanglement, i.e. entanglement due to antisymmetrization is clearly separated from ''genuine'' fermionic entanglement.\\
\\
This paper is organized as follows. In Section \ref{SN} we describe our notational and terminological conventions. Section \ref{SA} gives an introduction to relevant aspects of Arveson's concepts with a focus on its geometrical background. Afterwards we describe our generalization of his methods which is used in the subsequent proof of the main result of Section \ref{SA} on controling entanglement of states that are restricted to linear subspaces. In Section \ref{SB} we use these abstract results in order to define measures for bosonic and fermionic entanglement. Moreover, we establish the aforementioned relation between the entanglement measures for fermionic and distinguishable particles. We extend these results in Section \ref{SC} by introducing entanglement measures related to Schmidt and Slater numbers. For bipartite systems we derive similar quantitative relations between the measures for distinguishable and fermionic particles as in Section \ref{SB}. Moreover, we generalize a result of Johnston \cite{Johnston} concerning the quantification of Schmidt number entanglement which was formulated for bipartite finite dimensional quantum systems.

\section{Notation and Terminology}\label{SN}
Throughout this paper $\mathcal{H}$ will always denote a (possibly infinite dimensional) complex Hilbert space. The scalar-product on $\mathcal{H}$ is denoted by $\langle\cdot,\cdot\rangle$ and assumed to be linear in the first and anti-linear in the second component. General vectors of $\mathcal{H}$ are typically denoted by $\xi,\eta$, and $\zeta$ whereas for vectors of orthogonal systems we use $e$ and $f$. The algebra of bounded linear operators on $\mathcal{H}$ is denoted by $\mathcal{B}(\mathcal{H})$. A continuous linear functional $\varphi$ on $\mathcal{B}(\mathcal{H})$ is called \emph{normal} if there is a trace-class operator $\Phi$ on $\mathcal{H}$ such that
\[\varphi(x)=\mbox{tr}(\Phi x)\quad\mbox{for all }x\in\mathcal{B}(\mathcal{H})\,.\]
The set of normal linear functionals on $\mathcal{B}(\mathcal{H})$ is denoted by $\mathcal{B}(\mathcal{H})_{\ast}$. If $\dim\,\mathcal{H}=n<\infty$ every linear functional on $\mathcal{B}(\mathcal{H})$ is normal and the algebra $\mathcal{B}(\mathcal{H})$ can be identified with the set of complex $n\times n$-matrices which we denote by $M_n(\C)$. A \emph{state} on $\mathcal{B}(\mathcal{H})$ is a positive (hence continuous) linear functional of norm one. For a normal state $\varphi$ on $\mathcal{B}(\mathcal{H})$ its associated trace-class operator $\rho$ is referred to as \emph{density operator}. In this paper we use linear functionals instead of density operators since it leads to simpler notations.\\
\indent For vectors $\xi,\eta\in\mathcal{H}$ the \emph{rank-one-operator} $t_{\xi,\eta}$ is the linear operator on $\mathcal{H}$ defined by
$$t_{\xi,\eta}\zeta:=\langle\zeta,\eta\rangle\xi\quad\mbox{for }\zeta\in\mathcal{H}$$
and the corresponding (normal) \emph{vector functional} $\omega_{\xi,\eta}:\mathcal{B}(\mathcal{H})\rightarrow\C$ induced by $\xi,\eta$ is given by
$$\omega_{\xi,\eta}(x):=\langle x\xi,\eta\rangle=\mbox{tr}(t_{\xi,\eta}x)\quad\mbox{for }x\in\mathcal{B}(\mathcal{H})\,.$$
Occasionally we write $t_{\xi}$ and $\omega_{\xi}$ for $t_{\xi,\xi}$ and $\omega_{\xi,\xi}$ respectively.\\
\indent The unit ball and the unit sphere of a complex normed space $X$ are denoted by $B_{1,X}$ and $S_{1,X}$ respectively. A set $M\subset X$ is called \emph{balanced} if for all $x\in M$ and all numbers $\lambda\in\C$ with $|\lambda|=1$ we have $\lambda x\in M$. The convex hull of $M$ is denoted by $\mbox{co}(M)$ and we write $\overline{\mbox{co}}^{\norm{\cdot}}(M)$ for its norm closure. More generally, for an arbitrary topology $\mathcal{T}$ on $X$ the $\mathcal{T}$-closure of $\mbox{co}(M)$ is denoted by $\overline{\mbox{co}}^{\mathcal{T}}(M)$.\\
\indent We say that a set of continuous linear functionals $N$ on $X$ is \emph{separating} or \emph{separates points of $X$} if for every $x\in X$ the relation
$$\varphi(x)=0\quad\mbox{for all }\varphi\in N$$
already implies $x=0$. In particular, a subset $N\subset \mathcal{H}$ is separating if for every $\xi\in\mathcal{H}$ the condition
$$\langle\xi,\eta\rangle=0\quad\mbox{for all }\eta\in N$$
implies $\xi=0.$

\section{Abstract Theory}\label{SA}
In this section we give a brief introduction to Arveson's entanglement theory and generalize it according to our needs. Theorem \ref{TA}, the main result of this section, prepares grounds for comparing entanglement of distinguishable and fermionic particles in the sections \ref{SB} and \ref{SC}.
\begin{definition}\label{D4}
Let $V\subset \mathcal{H}$ be a balanced set of unit vectors that separates points of $\mathcal{H}$ and
$$K:=\overline{{\mathrm{co}}}^{\norm{\cdot}}\{\omega_{\xi,\eta}:\,\xi,\eta\in V\}\subset\mathcal{B}(\mathcal{H})_{\ast}\,.$$
Moreover, we define functions on $\mathcal{B}(\mathcal{H})$ and $\mathcal{B}(\mathcal{H})_{\ast}$ by
$$\norm{x}_{K}:=\sup\set{|\psi(x)|:\,\psi\in K}\quad \mbox{for }x\in\mathcal{B}(\mathcal{H})$$
and
$$q_K(\psi):=\inf\set{\alpha\geq 0:\,\psi\in\alpha K}\quad\mbox{for }\psi\in\mathcal{B}(\mathcal{H})_{\ast}$$
respectively. For the latter we use the convention that $q_K(\psi)=\infty$ whenever the set $\set{\alpha\geq 0:\,\psi\in\alpha K}$ is empty.
\end{definition}
\noindent The function $q_K$ is called \emph{Minkowski functional} associated to the set $K$ and quantifies the distance of a point $\psi\in\mathcal{B}(\mathcal{H})_{\ast}$ to $K$ along rays passing through the origin. In the case that $\mathcal{H}$ is a tensor product of Hilbert spaces appropriate choices of $V$ will allow to quantify multipartite entanglement in many different scenarios. 
\begin{proposition}\label{PA}
With $K$ and its associated functions $\norm{\cdot}_K$ and $q_K$ as defined above, we have the following properties:
\begin{enumerate}
\item\label{PA1}
The set $K$ is a norm-closed convex and balanced subset of the unit ball of $\mathcal{B}(\mathcal{H})_{\ast}$ that separates points of $\mathcal{B}(\mathcal{H})$. In particular, we have $0\in K$ and 
$$\norm{\psi}\leq q_K(\psi)\quad\mbox{for all }\psi\in\mathcal{B}(\mathcal{H})_{\ast}\,.$$
Moreover, $K$ equals the ''unit ball'' of $q_K$, i.e., 
\[K=\set{\psi\in\mathcal{B}(\mathcal{H})_{\ast}:\,q_K(\psi)\leq 1}\,.\]
\item\label{PA2}
The function $\norm{\cdot}_K$ is a norm on $\mathcal{B}(\mathcal{H})$ such that $\norm{x}_K\leq\norm{x}$ for all $x\in\mathcal{B}(\mathcal{H})$ and we have 
$$\norm{x}_K=\sup\set{|\omega_{\xi,\eta}(x)|:\,\xi,\eta\in V}=\sup\set{|\langle x\xi,\eta\rangle|:\,\xi,\eta\in V}\,.$$
\item\label{PA3}
The Minkowski functional $q_K$ has the alternative representation 
$$q_K(\psi)=\sup\set{|\psi(x)|:\,x\in\mathcal{B}(\mathcal{H}),\,\norm{x}_K\leq 1}\quad\mbox{for all }\psi\in\mathcal{B}(\mathcal{H})_{\ast}\,.$$
\end{enumerate}
\end{proposition}
\begin{proof}$\,$
\begin{enumerate}
\item 
The set $K$ is convex and norm-closed by definition. In order to show that $K$ separates points of $\mathcal{B}(\mathcal{H})$ it is sufficient to show that the set
$$\set{\omega_{\xi,\eta}:\,\xi,\eta\in V}\subset K$$
has this property. Therefore, suppose that $\omega_{\xi,\eta}(x)=0$ for all $\xi,\eta\in V$. Then
$$\langle x\xi,\eta\rangle=0\quad\mbox{for all }\xi,\eta\in V$$
and therefore $x\xi=0$ for all $\xi\in V$ since $V$ separates points of $\mathcal{H}$. For every $\zeta\in\mathcal{H}$ and $\xi\in V$ it follows
$$0=\langle x\xi,\zeta\rangle=\langle\xi,x^{\ast}\zeta\rangle\,.$$
Hence, $x^{\ast}\zeta=0$ for all $\zeta\in\mathcal{H}$ which implies
$$x=x^{\ast\ast}=0^{\ast}=0\,.$$
Since
$$\norm{\omega_{\xi,\eta}}=\norm{\xi}\cdot\norm{\eta}=1\mbox{ for all }\xi,\eta\in V$$
we have
$$K=\overline{\mbox{co}}^{\norm{\cdot}}\set{\omega_{\xi,\eta}:\,\xi,\eta\in V}\subset B_{1,\mathcal{B}(\mathcal{H})_{\ast}}\,.$$
It is a well-known fact (compare for example Section I.3 of \cite{Defant}) that for every $\varepsilon>0$ we have
$$K\subset \set{\varphi\in\mathcal{B}(\mathcal{H})_{\ast}:\,q_K(\varphi)\leq 1}\subset (1+\varepsilon)\cdot K$$
and using that $K$ is norm-closed it easily follows that $K$ equals the unit ball of $q_K$.
\item
For all $x\in\mathcal{B}(\mathcal{H})$ the mapping
$$\mathcal{B}(\mathcal{H})_{\ast}\ni\psi\mapsto |\psi(x)|\in\C$$
is convex and norm-continuous so that
\begin{eqnarray*}
\norm{x}_K&=&\sup\set{|\psi(x)|:\,\psi\in\overline{\mbox{co}}^{\norm{\cdot}}\set{\omega_{\xi,\eta}:\,\xi,\eta\in V}}\\
&=&\sup\set{|\omega_{\xi,\eta}(x)|:\,\xi,\eta\in V}\,.
\end{eqnarray*}
The remaining properties of $\norm{\cdot}_K$ are obvious.
\item
Using the concept of the \emph{polar} $K^{\circ}\subset\mathcal{B}(\mathcal{H})$ of $K$ with respect to the dual pair $(\mathcal{B}(\mathcal{H})_{\ast},\mathcal{B}(\mathcal{H}))$ (compare for example Section V.1 of \cite{Convay}) we have
\begin{eqnarray*}	
K^{\circ}&=&\set{x\in\mathcal{B}(\mathcal{H}):\,|\varphi(x)|\leq 1\mbox{ for all }\varphi\in K}\\
&=&\set{x\in\mathcal{B}(\mathcal{H}):\,\sup\set{|\varphi(x)|:\,\varphi\in K}\leq 1}\\
&=&\set{x\in\mathcal{B}(\mathcal{H}):\,\norm{x}_K\leq 1}\\
&=&B_{1,\norm{\cdot}_K}\,.
\end{eqnarray*}
The \emph{bipolar} $K^{\circ\circ}\subset\mathcal{B}(\mathcal{H})_{\ast}$ of $K$ is then given by
\begin{eqnarray*}	
K^{\circ\circ}&=&\set{\varphi\in\mathcal{B}(\mathcal{H})_{\ast}:\,|\varphi(x)|\leq 1\mbox{ for all }x\in K^{\circ}}\\
&=&\set{\varphi\in\mathcal{B}(\mathcal{H})_{\ast}:\,\sup\set{|\varphi(x)|:\,x\in K^{\circ}}\leq 1}\\
&=&\set{\varphi\in\mathcal{B}(\mathcal{H})_{\ast}:\,\sup\set{|\varphi(x)|:\,x\in \mathcal{B}(\mathcal{H}),\,\norm{x}_K\leq 1}\leq 1}\,.
\end{eqnarray*}
Using the fact that the $\sigma(\mathcal{B}(\mathcal{H})_{\ast},\mathcal{B}(\mathcal{H}))$-closure of a convex subset of $\mathcal{B}(\mathcal{H})_{\ast}$ equals its norm closure the Bipolar Theorem and $(i)$ imply
$$K\stackrel{(i)}{=}\overline{\mbox{co}}^{\norm{\cdot}}(K\cup\set{0})=\overline{\mbox{co}}^{\sigma(\mathcal{B}(\mathcal{H})_{\ast},\mathcal{B}(\mathcal{H}))}(K\cup\set{0})=K^{\circ\circ}\,.$$
Hence $K$ equals the unit ball of the function
$$f(\varphi):=\sup\set{|\varphi(x)|:\,x\in \mathcal{B}(\mathcal{H}),\,\norm{x}_K\leq 1}\quad\mbox{for }\varphi\in\mathcal{B}(\mathcal{\varphi})_{\ast}\,.$$
Therefore, since the functions $q_K$ and $f$ possess the same unit ball and are homogeneous for positive scalars it follows $q_K=f$.
\end{enumerate}
\end{proof}
\noindent Note that from the representation of $q_K$ according to Proposition \ref{PA}.(\ref{PA3}) it easily follows that $q_K$ satisfies the triangle inequality. Hence, $q_K$ has all characteristic properties of a norm except that it may attain the value ''$\infty$''. We may therefore speak of a \emph{generalized norm}.\\ 
\noindent The above construction slightly differs from Arveson's approach. In his publication \cite{Arveson} the unital $\mbox{C}^{\ast}$-algebra 
$$\mathcal{A}:=\mathcal{K}(\mathcal{H})+\C\mathbbmss{1}$$
is used in order to define a function $E:\mathcal{B}(\mathcal{H})_{\ast}\rightarrow[0,\infty]$ given by
\[E(\psi):=\sup\set{|\psi(x)|:\,x\in\mathcal{A},\,\norm{x}_K\leq 1}\quad\mbox{for }\psi\in\mathcal{B}(\mathcal{H})_{\ast}\,.\]
Hence, comparing $q_K$ with $E$ we see that for $q_K$ the supremum according to Proposition \ref{PA}.(\ref{PA3}) runs over $\mathcal{B}(\mathcal{H})$ whereas in the case of $E$ the supremum is restricted to the $\mbox{C}^{\ast}$-algebra $\mathcal{A}$. As a matter of fact both functions $E$ and $q_K$ can detect separability equally well.
\begin{theorem}\label{TB}
Let $\varphi$ be a normal state on $\mathcal{B}(\mathcal{H})$. Then the following conditions are equivalent:
\begin{enumerate}
\item 
$\varphi\in\overline{\mathrm{co}}^{\norm{\cdot}}\{\omega_{\xi}:\,\xi\in V\}$.
\item
$E(\varphi)=1$.
\item
$q_K(\varphi)=1$.
\end{enumerate}
\end{theorem}
\begin{proof}
The equivalence of (i) and (ii) has been proven by Arveson (compare Theorem 6.2 in \cite{Arveson}). A careful inspection of his proof reveals that the only relevant property of the algebra $\mathcal{A}$ is the fact that it containes the unit operator (which guarantees the set of states on $\mathcal{A}$ to be compact in its $\sigma^{\ast}$-topology) as well as the set of compact operators on $\mathcal{H}$. Both conditions are met by the $\mbox{C}^{\ast}$-algebra $\mathcal{B}(\mathcal{H})$ as well so that the equivalence can be proven literary.
\end{proof}
\noindent Note that according to the proof of Theorem \ref{TB} \emph{any} $\mbox{C}^{\ast}$-subalgebra of $\mathcal{B}(\mathcal{H})$ containing the set of compact operators and the unit operator induces an (abstract) entanglement measure with essentially the same structure as $E$ and $q_K$. Therefore, it is justified to speak of a generalization of Arveson's result.\\
\indent The equivalence of (i) and (ii) in the above theorem is the basis of Arveson's entanglement quantification. However, due to the fact that these statements are equivalent to (iii) as well both functions $E$ and $q_K$ may serve as entanglement measures. We will use the latter since our main result Theorem \ref{TA} needs to be formulated in terms of $q_K$.\\
\indent Now suppose we are given a nonzero orthogonal projection $p\in\mathcal{B}(\mathcal{H})$ and a normal state $\varphi$ on $\mathcal{B}(\mathcal{H})$ such that
$$\varphi(x)=\varphi(pxp)\quad\mbox{for all }x\in\mathcal{B}(\mathcal{H})\,.$$
Since $\varphi(p)=\varphi(p\mathbbmss{1}p)=\varphi(\mathbbmss{1})=1$ we may identify $\varphi$ with a normal state on the algebra $\mathcal{B}(p\mathcal{H})$. On the other hand, projecting the set $V$ to the subspace $p\mathcal{H}$ (along with some further operations) gives rise to a convex and balanced subset $K_p\subset\mathcal{B}(p\mathcal{H})_{\ast}$ according to Definition \ref{D4}. Therefore, it is a natural question to ask as to how the values $q_K(\varphi)$ and $q_{K_p}(\varphi)$ are related to each other. Theorem \ref{TA} gives a very general answer to that question. In the remainder of this section we will discuss it in detail and to this end we introduce the following notation.
\begin{definition}\label{D5}
Let $0\neq p\in\mathcal{B}(\mathcal{H})$ be an orthogonal projection. For a number $\lambda\geq 1$ we define subsets of $p\mathcal{H}$ and $\mathcal{B}(p\mathcal{H})_{\ast}$ by
$$V_{p,\lambda}:=\lambda pV\cap S_{1,\mathcal{H}}$$
and
$$K_{p,\lambda}:=\overline{\mathrm{co}}^{\norm{\cdot}}\{\omega_{\xi,\eta}:\,\xi,\eta\in V_{p,\lambda}\}$$
respectively. If $V_{p,\lambda}$ satisfies the condition that for all $\xi\in V$ with $p\xi\neq 0$ the relation 
$$\frac{p\xi}{\norm{p\xi}}\in V_{p,\lambda}\,,$$
is fulfilled we say that the projection $p$ is \emph{compatible} with the pair $(V,\lambda)$.
\end{definition}
\noindent Note that the set $V_{p,\lambda}$ as defined above may be empty. In particular, it is not necessarily a separating subset of $p\mathcal{H}$. However, it is always balanced since this property is inherited from $V$.\\
\indent We will see in the sections \ref{SB} and \ref{SC} that the conditions formulated in the following theorem are met by important physical scenarios although this may seem very unlikely on the first glance.
\begin{theorem}\label{TA}
For a nonzero orthogonal projection $p\in\mathcal{B}(\mathcal{H})$ and a number $\lambda\geq 1$ let $V_{p,\lambda}$ and $K_{p,\lambda}$ be the sets introduced in Definition \ref{D5}.
\begin{enumerate}
\item\label{TA1}
Suppose that $V_{p,\lambda}$ separates points of $p\mathcal{H}$ and that the mapping 
$$\mathcal{B}(\mathcal{H})\ni y\mapsto pyp\in\mathcal{B}(\mathcal{H})$$
is a $\norm{\cdot}_{K}$-contraction. Identifying $\mathcal{B}(p\mathcal{H})$ with the subalgebra $p\mathcal{B}(\mathcal{H})p\subset\mathcal{B}(\mathcal{H})$ we have the following inequalities:
\begin{enumerate}
\item 
$\norm{x}_{K_{p,\lambda}}\leq\lambda^2\cdot\norm{x}_{K}\mbox{ for all }x\in\mathcal{B}(p\mathcal{H})\,.$
\item\label{TA12}
$q_{K_{p,\lambda}}(\psi)\geq\frac{1}{\lambda^2}\cdot q_{K}(\psi)\mbox{ for all }\psi\in\mathcal{B}(p\mathcal{H})_{\ast}\,.$
\end{enumerate}
\item\label{TA2}
If $p$ is compatible with the pair $(V,\lambda)$ then $V_{p,\lambda}$ separates points of $p\mathcal{H}$. Moreover, defining
$$\mu:=\sup\set{r\geq 1:\,rpV\in B_{1,\mathcal{H}}}$$
we have:
\begin{enumerate}
\item 
$\norm{x}_{K_{p,\lambda}}\geq\mu^2\cdot\norm{x}_{K}\mbox{ for all }x\in\mathcal{B}(p\mathcal{H})$\,.
\item\label{TA22}
$q_{K_{p,\lambda}}(\psi)\leq\frac{1}{\mu^2}\cdot q_{K}(\psi)\mbox{ for all }\psi\in\mathcal{B}(\mathcal{H})_{\ast}\,.$
\end{enumerate}
\end{enumerate}
\end{theorem}
\begin{proof}$\,$
\begin{enumerate}
\item
By the construction of $V_{p,\lambda}$ we obtain
\begin{eqnarray*}
\norm{x}_{K_{p,\lambda}}&=&\sup\set{\left|\left\langle x\xi,\eta\right\rangle\right|:\,{\xi,\eta\in V_{p,\lambda}}}\\
&=&\lambda^2\cdot\sup\set{\left|\left\langle \underbrace{pxp}_{=x}\xi,\eta\right\rangle\right|:\,\xi,\eta\in V,\,\norm{\lambda p\xi}=\norm{\lambda p\eta}=1}\\
&\leq&\lambda^2\cdot\sup\set{\left|\left\langle x\xi,\eta\right\rangle\right|:\,{\xi,\eta\in V}}=\lambda^2\cdot\norm{x}_{K}
\end{eqnarray*} 
which proves inequality (a) (note that we did not use the contraction condition imposed on $\norm{\cdot}_K$).\\
\indent In order to prove (b) we first consider the case $q_{K}(\psi)=\infty$. Then by Proposition \ref{PA}.(\ref{PA3}) for every $C>0$ there is an element $x\in \mathcal{B}(\mathcal{H})$ with $\norm{x}_{K}\leq 1$ such that $|\psi(x)|\geq C$. Inequality (a) and the assumption about $\norm{\cdot}_{K}$ then imply
$$\norm{pxp}_{K_{p,\lambda}}\leq\lambda^2\cdot\norm{pxp}_{K}\leq \lambda^2\cdot\norm{x}_{K}\leq\lambda^2\,.$$
Using $\psi(x)=\psi(pxp)$ for all $x\in\mathcal{B}(\mathcal{H})$ a further application of Proposition \ref{PA}.(\ref{PA3}) to $q_{K_{p,\lambda}}$ gives
$$q_{K_{p,\lambda}}(\psi)\geq\left|\psi\left(\frac{pxp}{\lambda^2}\right)\right|=\frac{1}{\lambda^2}|\psi(x)|\geq\frac{C}{\lambda^2}\,,$$
that is $q_{K_{p,\lambda}}(\psi)=\infty$.\\
\indent On the other hand, if $q_{K}(\psi)<\infty$, Proposition \ref{PA}.(\ref{PA3}) allows to choose for every $\varepsilon>0$ an element $x\in \mathcal{B}(\mathcal{H})$ with $\norm{x}_{K}\leq 1$ such that
$$q_{K}(\psi)-\varepsilon<|\psi(x)|\,.$$
By the same computation as above we have $\norm{pxp}_{K_{p,\lambda}}\leq\lambda^2$ and it follows
$$|\psi(x)|=\lambda^2\cdot\left|\psi\left(\frac{pxp}{\lambda^2}\right)\right|\leq\lambda^2\cdot\sup\set{|\psi(y)|:\,y\in \mathcal{B}(p\mathcal{H}),\,\norm{y}_{K_{p}}\leq 1}=\lambda^2\cdot q_{K_{p,\lambda}}({\psi})$$
and therefore 
$$\frac{1}{\lambda^2}\cdot q_{K}(\psi)<\frac{\varepsilon}{\lambda^2}+q_{K_{p,\lambda}}(\psi)\,.$$
Hence, $\frac{1}{\lambda^2}\cdot q_{K}(\psi)\leq q_{K_{p,\lambda}}(\psi)$.
\item
Let $\xi\in p\mathcal{H}$ such that $\langle\xi,\eta\rangle=0$ for all $\eta\in V_{p,\lambda}$. For $\zeta\in V$ with $p\zeta\neq 0$ it follows from our assumption
$$\langle\xi,\zeta\rangle=\langle p\xi,\zeta\rangle=\langle\xi,p\zeta\rangle=\norm{p\zeta}\cdot\left\langle\xi,\frac{p\zeta}{\norm{p\zeta}}\right\rangle=0$$
and for $p\zeta=0$ we find analogously that $\langle\xi,\zeta\rangle=0$. Since $V$ separates points of $\mathcal{H}$ this implies $\xi=0$. Hence, $V_{p,\lambda}$ separates points of $p\mathcal{H}$.\\
In order to prove (a) we may assume that $x\neq 0$. Using Proposition \ref{PA}.(\ref{PA2}) for $\norm{x}_{K}>\varepsilon>0$ we can choose $\xi,\eta\in V$ such that
$$\norm{x}_{K}-\varepsilon<\left|\left\langle x\xi,\eta\right\rangle\right|\,.$$
Then we must have
$$0\neq|\langle x\xi,\eta\rangle|=|\langle xp\xi,p\eta\rangle|$$
and, in particular, $p\xi,p\eta\neq 0$. By the choice of $\mu$ and the condition imposed on $V_{p,\lambda}$ we obtain
\begin{eqnarray*}
\norm{x}_K-\varepsilon&<&\left|\left\langle x\xi,\eta\right\rangle\right|
=\frac{1}{\mu^2}\cdot\underbrace{\norm{\mu p\xi}}_{\leq 1}\cdot\underbrace{\norm{\mu p\eta}}_{\leq 1}\cdot\left|\left\langle x\frac{p\xi}{\norm{p\xi}},\frac{p\eta}{\norm{p\eta}}\right\rangle\right|\\
&\leq&\frac{1}{\mu^2}\cdot\sup\set{|\langle x\xi',\eta'\rangle|:\,\xi',\eta'\in V_{p,\lambda}}=\frac{1}{\mu^2}\norm{x}_{K_{p,\lambda}}\,.
\end{eqnarray*}
It follows
$$\mu^2\cdot\norm{x}_{K}<\mu^2\cdot\varepsilon+\norm{x}_{K_{p,\lambda}}\,;$$
that is $\mu^2\cdot\norm{x}_{K}\leq\norm{x}_{K_{p,\lambda}}$. By Proposition \ref{PA}.(\ref{PA3}) the inequality (b) now follows from (a) via
\begin{eqnarray*}
q_{K_{p,\lambda}}(\psi)&=&\sup\set{|\psi(x)|:\,x\in \mathcal{B}(p\mathcal{H}),\,\norm{x}_{K_{p,\lambda}}\leq 1}\\
&\stackrel{\mbox{\footnotesize{(a)}}}{\leq}&\frac{1}{\mu^2}\cdot\sup\set{|\psi(x)|:\,{x\in \mathcal{B}(p\mathcal{H})\,,\norm{x}_{K}\leq 1}}\\
&\leq&\frac{1}{\mu^2}\cdot\sup\set{|\psi(x)|:\,{x\in \mathcal{B}(\mathcal{H})\,,\norm{x}_{K}\leq 1}}\\
&=&\frac{1}{\mu^2}\cdot q_{K}(\psi)\,.
\end{eqnarray*}
\end{enumerate}
\end{proof}
\noindent Note that if the conditions of Theorem \ref{TA}.(\ref{TA1}) and (\ref{TA2}) are both satisfied we have the chains of inequalities
$$\mu^2\cdot\norm{x}_K\leq\norm{x}_{K_{p,\lambda}}\leq\lambda^2\cdot\norm{x}_K\quad\mbox{for all }x\in\mathcal{B}(p\mathcal{H})$$
and
$$\frac{1}{\lambda^2}\cdot q_K(\psi)\leq q_{K_{p,\lambda}}(\psi)\leq\frac{1}{\mu^2}q_K(\psi)\quad\mbox{for all }\psi\in\mathcal{B}(p\mathcal{H})_{\ast}\,.$$

\section{Quantifying Multipartite Entanglement of Indistinguishable Particles}\label{SB}
For a number $k\in\mathbbmss{N}$ we denote the $k$-fold Hilbert space tensor product of $\mathcal{H}$ with itself by $\mathcal{H}^{\otimes k}$. The following notion is the starting point for the quantification of entanglement of distinguishable particles.
\begin{definition}\label{D1}
We denote the set of unit product vectors of $\mathcal{H}^{\otimes k}$ by $V^{\otimes k}$, i.e.
$$V^{{\otimes k}}:=\set{\eta_1\otimes...\otimes\eta_k:\,\eta_i\in\mathcal{H},\,\norm{\eta_i}=1}\,.$$
A normal state $\varphi$ on $\mathcal{B}(\mathcal{H}^{\otimes k})$ is called \emph{separable} if
$$\varphi\in\overline{\mathrm{co}}^{\norm{\cdot}}\set{\omega_{\xi}:\,\xi\in V^{\otimes k}}\,.$$
Otherwise, it is called \emph{entangled}.
\end{definition}
\noindent Noting that the set $V^{{\otimes k}}$ is a balanced subset of the unit sphere that separates points of $\mathcal{H}^{\otimes k}$ we obtain a necessary and sufficient separability criterion for normal states.
\begin{corollary}\label{CA}
Let $K^{\otimes k}$ denote the norm-closed convex hull of the set
$$\left\{\omega_{\xi,\eta}:\,\xi,\eta\in V^{\otimes k}\right\}\,.$$
Then every normal state $\varphi$ on $\mathcal{B}(\mathcal{H}^{\otimes k})$ satisfies $q_{K^{\otimes k}}(\varphi)\geq 1$ and $q_{K^{\otimes k}}(\varphi)=1$ if and only if $\varphi$ is separable. 
\end{corollary}
\begin{proof}
The claim immediately follows from Proposition \ref{PA}.(\ref{PA1}) and Theorem \ref{TB}.
\end{proof}
\noindent It is shown in \cite{Arveson} that for finite dimensional Hilbert spaces Corollary \ref{CA} is the \emph{greatest cross-norm criterion} introduced by Rudolph \cite{Rudolph0,Rudolph,Rudolph2}. We will encounter such situations in the examples \ref{E1}-\ref{E2} below.\\   
In order to define separability for states of indistinguishable particles some conceptual and notational preparation is required.\\
\begin{definition}
Let $\pi\in S_k$ be a permutation of $k$ points. We denote the associated permutation operator on $\mathcal{H}^{\otimes k}$ by $U_{\pi}$. The orthogonal projections onto the bosonic and fermionic subspaces are given by
\begin{eqnarray*}
P_{+}:=\frac{1}{k!}\sum_{\pi\in S_k}{U_{\pi}}\quad\mbox{and}\quad P_{-}:=\frac{1}{k!}\sum_{\pi\in S_k}{\mathrm{sign}(\pi)U_{\pi}}\\
\end{eqnarray*}
and we write
$$\mathcal{H}^{\vee k}:=P_+\mathcal{H}^{\otimes k}\quad\mbox{and}\quad\mathcal{H}^{\wedge k}:=P_-\mathcal{H}^{\otimes k}\,.$$
Furthermore, for $\eta_1,...,\eta_k\in\mathcal{H}$ we define
$$\bigvee_{i=1}^k{\eta_i}:=\eta_1\vee...\vee\eta_k:=P_+\eta_1\otimes...\otimes\eta_k$$
and
$$\bigwedge_{i=1}^k{\eta_i}:=\eta_1\wedge...\wedge\eta_k:=\sqrt{k!}P_-\eta_1\otimes...\otimes\eta_k\,.$$
\end{definition}
\noindent It is a well-known fact that for all $\eta_i\in\mathcal{H}$ and $\pi\in S_k$ we have
$$U_{\pi}\left(\eta_1\vee...\vee\eta_k\right)=\eta_1\vee...\vee\eta_k\quad\mbox{and}\quad U_{\pi}\left(\eta_1\wedge...\wedge\eta_k\right)=\mbox{sign}(\pi)\eta_1\wedge...\wedge\eta_k\,.$$
Using the latter relation it is easy to prove that
$$\langle\xi_1\wedge...\wedge\xi_k,\eta_1\wedge...\wedge\eta_k\rangle=\det\left(\begin{array}{ccc}\langle\xi_1,\eta_1\rangle&\cdots&\langle\xi_k,\eta_1\rangle\\\vdots&&\vdots\\\langle\xi_1,\eta_k\rangle&\cdots&\langle\xi_k,\eta_k\rangle\end{array}\right)$$ 
for all $\xi_i,\eta_i\in\mathcal{H}$. In particular, this implies
$$\eta_1\wedge...\wedge\eta_k=0$$
whenever the vectors $\eta_1,...,\eta_k$ are linearly dependent.
\begin{lemma}\label{LA}
Let $\eta_1,...,\eta_k\in\mathcal{H}$.
\begin{enumerate}
\item\label{LA1}
There exists an orthogonal system $\{\eta_1',...,\eta_k'\}\subset\mathcal{H}$ such that
$$\bigwedge_{i=1}^k{\eta_i}=\bigwedge_{i=1}^k{\eta_i'}$$
and it can be chosen in such a way that $\norm{\eta_i'}\leq\norm{\eta_i}$ for all $1\leq i\leq k$.
\item\label{LA2}
We have the estimation
$$\norm{\bigwedge_{i=1}^k{\eta_i}}\leq\prod_{i=1}^k{\norm{\eta_i}}\,.$$
and equality is given if and only if $\{\eta_1,...,\eta_k\}$ is an orthogonal system.
\end{enumerate}
\end{lemma}
\begin{proof}$\,$
\begin{enumerate}
\item
For every $i\in\{2,...,k\}$ there are elements $\eta_i^{(1)},\eta_i^{(2)}\in\mathcal{H}$ with $\eta_i^{(1)}\in\mathbbmss{C}\cdot\eta_1$ and $\eta_i^{(2)}\in(\mathbbmss{C}\cdot\eta_1)^{\bot}$ such that $\eta_i=\eta_i^{(1)}+\eta_i^{(2)}$. It follows
$$\bigwedge_{i=1}^k{\eta_i}=\eta_1\wedge\bigwedge_{i=2}^k{\left(\eta_i^{(1)}+\eta_i^{(2)}\right)}=\sum_{j_2=1}^2{...\sum_{j_k=1}^2{\underbrace{\eta_1\wedge\bigwedge_{i=2}^k{\eta_i^{(j_i)}}}_{=0\mbox{\footnotesize{ if }}j_i=1}}}=\eta_1\wedge\bigwedge_{i=2}^k{\eta_i^{(2)}}$$
and the vectors $\eta_i^{(2)}$ are orthogonal to $\eta_1$. If we set $\eta_1':=\eta_1$ and apply the same procedure to $\bigwedge_{i=2}^{k}{\eta_i^{(2)}}$ we iteratively obtain the vectors $\eta_1',...,\eta_k'$ as claimed. Due to $$\norm{\eta_2'}^2=\norm{\eta_2^{(2)}}^2=\norm{\eta_2}^2-\norm{\eta_2^{(1)}}^2\leq\norm{\eta_2}^2$$
one also iteratively validates the inequalities $\norm{\eta_i'}\leq\norm{\eta_i}$.
\item
By (i) there is an orthogonal system $\left\{\eta_1',...,\eta_k'\right\}\subset\mathcal{H}$ with $\bigwedge_{i=1}^k{\eta_i}=\bigwedge_{i=1}^k{\eta_i'}$ and $\norm{\eta_i'}\leq\norm{\eta_i}$. It follows
\begin{eqnarray*}
\norm{\eta_1\wedge...\wedge\eta_k}^2&=&\norm{\eta_1'\wedge...\wedge\eta_k'}^2=\det\left(\begin{array}{ccc}\langle\eta_1',\eta_1'\rangle&\cdots&\langle\eta_k',\eta_1'\rangle\\\vdots&&\vdots\\\langle\eta_1',\eta_k'\rangle&\cdots&\langle\eta_k',\eta_k'\rangle\end{array}\right)\\
&=&\norm{\eta_1'}^2\cdot...\cdot\norm{\eta_k'}^2\leq\norm{\eta_1}^2\cdot...\cdot\norm{\eta_k}^2
\end{eqnarray*}
which proves the inequality. Moreover, we see that
$$\norm{\eta_1\wedge...\wedge\eta_k}=\norm{\eta_1}\cdot...\cdot\norm{\eta_k}$$
if $\{\eta_1,...,\eta_k\}$ happens to be an orthogonal system. Conversely, if this is not the case then the procedure described in (i) allows to choose the orthogonal system $\{\eta_1',...,\eta_k'\}$ in such a way that at least one of the inequalities $\norm{\eta_i'}\leq\norm{\eta_i}$ is strict, so that the above inequality is strict as well.
\end{enumerate}
\end{proof}
\noindent In particular, Lemma \ref{LA} implies that the space $\mathcal{H}^{\wedge k}$ does not contain any nonzero product vectors. Therefore, the ''classical'' notion of entanglement of fermionic particles does not apply to this case.\\
\indent Whereas for fermionic particles an accepted definition of entanglement is found \cite{Eckert,Grabowski,Levay,Plastino,Schliemann,Zander} there is a still ongoing debate as to how entanglement of bosonic particles should be defined \cite{Eckert,Ghi02,Ghi04,Ghi05,Grabowski,Paskauskas,Li}. In the following we will concentrate on the definition of bosonic entanglement to which our techniques apply. 
\begin{definition}\label{D2}$\,$
\begin{enumerate}
\item\label{D21} 
The set of symmetric product vectors of $\mathcal{H}^{\vee k}$ is denoted by
$$V^{\vee k}:=\set{\eta\otimes...\otimes\eta:\,\eta\in\mathcal{H},\,\norm{\eta}=1}\,.$$
We say that a normal state $\varphi$ on $\mathcal{B}\left(\mathcal{H}^{\vee k}\right)$ is \emph{bosonic separable} if
$$\varphi\in\overline{\mathrm{co}}^{\norm{\cdot}}\set{\omega_{\xi}:\,\xi\in V^{\vee k}}\,.$$
\item\label{D22}
Similarly, we define 
$$V^{\wedge k}:=\set{\eta_1\wedge...\wedge\eta_k:\,\eta_i\in\mathcal{H},\,\norm{\eta_1\wedge...\wedge\eta_k}=1}$$
and call a normal state $\varphi$ on $\mathcal{B}\left(\mathcal{H}^{\wedge k}\right)$ \emph{fermionic separable} if
$$\varphi\in\overline{\mathrm{co}}^{\norm{\cdot}}\set{\omega_{\xi}:\,\xi\in V^{\wedge k}}\,.$$
\end{enumerate}
\end{definition}
\noindent Comparing the definitions \ref{D1}, \ref{D2}.(\ref{D21}) and \ref{D2}.(\ref{D22}) one recognize that these are just special cases of the more general concept discussed in Section \ref{SA}. In addition, the sets $V^{\otimes k}$, $V^{\vee k}$ and $V^{\wedge k}$ fulfill the following geometric relations which are of significant importance for entanglement quantification. 
\begin{proposition}\label{PC}
According to Theorem \ref{TA}.(\ref{TA2}) let us define the numbers
$$\mu_{\pm}:=\sup\left\{r\geq 1:\,rP_{\pm}V^{\otimes k}\subset B_{1,\mathcal{H}^{\otimes k}}\right\}\,.$$
\begin{enumerate}
\item\label{PC1} 
We have $\mu_+=1$ and $\mu_-=\sqrt{k!}$ as well as
$$ V^{\vee k}=\left(P_{+}V^{\otimes k}\right)\cap S_{1,\mathcal{H}^{\otimes k}}\quad\mbox{and}\quad V^{\wedge k}=\left(\sqrt{k!}P_{-}V^{\otimes k}\right)\cap S_{1,\mathcal{H}^{\otimes k}}\,.$$
\item\label{PC2}
The projection $P_-$ is compatible with the pair $(V^{\otimes k},\mu_-)$ in the sense of Definition \ref{D5} whereas $P_+$ is not compatible with $(V^{\otimes k},\mu_+)$.
\item\label{PC3}
The mappings
$$\mathcal{B}\left(\mathcal{H}^{\otimes k}\right)\ni x\mapsto P_{\pm}xP_{\pm}\in\mathcal{B}\left(\mathcal{H}^{\otimes k}\right)$$
are $\norm{\cdot}_{K^{\otimes k}}$-contractions (compare Definition \ref{D4} and Corollary \ref{CA}).
\end{enumerate}
\end{proposition}
\begin{proof}$\,$
\begin{enumerate}
\item 
We begin by proving the statements about $\mu_+$ and $V^{\vee k}$. Clearly, we have $V^{\vee k}\subset V^{\otimes k}$ and every vector $\eta\in V^{\vee k}$ is a fixed-point of $P_+$. Hence, $\mu_+=1$ and $V^{\vee k}\subset (P_+V^{\otimes k})\cap S_{1,\mathcal{H}^{\otimes k}}$. Conversely, for $\eta=\eta_1\otimes...\otimes\eta_k\in V^{\otimes k}$ we have
\[P_+\eta=\frac{1}{k!}\sum_{\pi\in S_k}U_{\pi}\eta=\sum_{\pi\in S_k}{\frac{1}{k!}\cdot\eta_{\pi(1)}\otimes...\otimes\eta_{\pi(k)}}\,;\]
that is $P_+\eta$ is a convex combination of unit vectors. Since the unit ball of a Hilbert space is strictly convex it follows that $P_+\eta$ is a unit vector if and only if $\eta=U_{\pi}\eta$ for all $\pi\in S_k$ which implies $\eta\in V^{\vee k}$.\\
\noindent Let us now consider the fermionic analogue. Using Lemma \ref{LA}.(\ref{LA2}) we find for $\eta_1\otimes...\otimes\eta_k\in V^{\otimes k}$ that
$$\norm{\sqrt{k!}P_-\eta_1\otimes...\otimes\eta_k}=\norm{\eta_1\wedge...\wedge\eta_k}\leq\norm{\eta_1}\cdot...\cdot\norm{\eta_k}=1$$
and we have equality if and only if $\{\eta_1,...,\eta_k\}$ is an orthogonal system. This proves $\mu_-=\sqrt{k!}$. On the other hand, for a given unit vector $\eta_1\wedge...\wedge\eta_k\in V^{\wedge k}$ Lemma \ref{LA}.(\ref{LA1}) allows to choose an orthogonal system $\{\eta_1',...,\eta_k'\}\subset\mathcal{H}$ with
$$\eta_1\wedge...\wedge\eta_k=\eta_1'\wedge...\wedge\eta_k'\,.$$
Due to
$$1=\norm{\eta_1\wedge...\wedge\eta_k}=\norm{\eta_1'\wedge...\wedge\eta_k'}=\norm{\eta_1'}\cdot...\cdot\norm{\eta_k'}$$
and using that $\left\{\frac{\eta_1'}{\norm{\eta_1'}},...,\frac{\eta_k'}{\norm{\eta_k'}}\right\}$ is an orthonormal system it follows
$$\eta_1\wedge...\wedge\eta_k=\frac{\eta_1'}{\norm{\eta_1'}}\wedge...\wedge\frac{\eta_k'}{\norm{\eta_k'}}=\sqrt{k!}P_-\frac{\eta_1'}{\norm{\eta_1'}}\otimes..\otimes\frac{\eta_k'}{\norm{\eta_k'}}\in\sqrt{k!}P_-V^{\otimes k}$$
so that $V^{\wedge k}\subset\left(\sqrt{k!}P_-V^{\otimes k}\right)\cap S_{1,\mathcal{H}}$. The converse inclusion is trivial.
\item
If $\eta=\eta_1\otimes...\otimes\eta_k\in V^{\otimes k}$ with $P_-\eta\neq 0$, then 
$$\frac{P_-\eta}{\norm{P_-\eta}}=\frac{\eta_1\wedge...\wedge\eta_k}{\norm{\eta_1\wedge...\wedge\eta_k}}\in V^{\wedge k}$$
by definition of $V^{\wedge k}$ so that $P_-$ is compatible with $(V^{\otimes k},\mu_-)$.\\
In order to see that $P_+$ is not compatible with $(V^{\otimes k},\mu_+)$ it is sufficient to consider the case $k=2$. For example, choosing two linearly independent unit vectors $\xi,\eta\in\mathcal{H}$ it is clear that $P_+(\xi\otimes \eta)=\frac{1}{2}\left(\xi\otimes \eta+\eta\otimes \xi\right)$ is nonzero but no product vector. In particular, if rescaled to unit length it cannot belong to $V^{\vee 2}$.
\item
For every element $\pi\in S_k$ the permutation operator $U_{\pi}$ is a bijection of $V^{\otimes k}$ onto itself. Thus, for all $\pi,\sigma\in S_k$ the mapping
$$\mathcal{B}\left(\mathcal{H}^{\otimes k}\right)\ni x\mapsto U_{\pi}xU_{\sigma}\in\mathcal{B}\left(\mathcal{H}^{\otimes k}\right)$$
is $\norm{\cdot}_{K^{\otimes k}}$-isometric since for every $x\in\mathcal{B}\left(\mathcal{H}^{\otimes k}\right)$ we have
\begin{eqnarray*}
\norm{U_{\pi}xU_{\sigma}}_{K^{\otimes k}}&=&\sup\left\{\left|\left\langle U_{\pi}xU_{\sigma}\xi,\eta\right\rangle\right|:\,\xi,\eta\in V^{\otimes k}\right\}\\
&=&\sup\left\{\left|\left\langle xU_{\sigma}\xi,U_{\pi}^{\ast}\eta\right\rangle\right|:\,\xi,\eta\in V^{\otimes k}\right\}\\
&=&\sup\left\{\left|\left\langle x\xi,\eta\right\rangle\right|:\,\xi,\eta\in V^{\otimes k}\right\}\\
&=&\norm{x}_{K^{\otimes k}}\,.
\end{eqnarray*}
It follows
\begin{eqnarray*}
\norm{P_{\pm}xP_{\pm}}_{K^{\otimes k}}\leq\frac{1}{(k!)^2}\sum_{\pi,\sigma\in S_k}{\underbrace{\norm{U_{\pi}xU_{\sigma}}_{K^{\otimes k}}}_{=\norm{x}_{K^{\otimes k}}}}=\norm{x}_{K^{\otimes k}}\,.
\end{eqnarray*}
\end{enumerate}
\end{proof}
\noindent Using Proposition \ref{PC} and by analogy with Corollary \ref{CA} we arrive at the following result concerning entanglement quantification of indistinguishable particles. 
\begin{proposition}\label{PD}
Let $K^{\vee k}$ and $K^{\wedge k}$ denote the norm-closed convex hulls of the sets 
$$\left\{\omega_{\xi}:\,\xi\in V^{\vee k}\right\}\quad\mbox{and}\quad \left\{\omega_{\xi}:\,\xi\in V^{\wedge k}\right\}$$
respectively. Then for normal states $\varphi_+\in\mathcal{B}\left(\mathcal{H}^{\vee k}\right)_{\ast}$ and $\varphi_-\in\mathcal{B}\left(\mathcal{H}^{\wedge k}\right)_{\ast}$ we have the following equivalences:
\begin{enumerate}
\item\label{PD1} 
We have $q_{K^{\vee k}}(\varphi_+)\geq 1$ and $q_{K^{\vee k}}(\varphi_+)= 1$ if and only if $\varphi_+$ is bosonic separable.
\item\label{PD2}
We have $q_{K^{\wedge k}}(\varphi_-)\geq 1$ and $q_{K^{\wedge k}}(\varphi_-)= 1$ if and only if $\varphi_-$ is fermionic separable.
\end{enumerate}
Furthermore, considering $\varphi_-$ as a normal state on $\mathcal{B}\left(\mathcal{H}^{\otimes k}\right)$ via 
$$\mathcal{B}\left(\mathcal{H}^{\wedge k}\right)\cong P_-\mathcal{B}\left(\mathcal{H}^{\otimes k}\right)P_-\subset\mathcal{B}\left(\mathcal{H}^{\otimes k}\right)$$
we have
$$q_{K^{\wedge k}}(\varphi_-)=\frac{1}{k!}\cdot q_{K^{\otimes k}}(\varphi_-)\,.$$
In particular, $q_{K^{\otimes k}}(\varphi_-)\geq k!$ and $q_{K^{\otimes k}}(\varphi_-)=k!$ if and only if $\varphi_-$ is fermionic separable.
\end{proposition}
\begin{proof}
Both sets $V^{\wedge k}$ and $V^{\vee k}$ are balanced subsets of the corresponding unit spheres. Moreover, since $P_-$ is compatible with $(V^{\otimes k},\mu_-)$ by Proposition \ref{PC}.(\ref{PC2}) the set $V^{\wedge k}$ separates points of $\mathcal{H}^{\wedge k}$ according to Theorem \ref{TA}.(\ref{TA2}). It is shown in \cite{OpenQuantumSystems} (compare Lemma 1.1 of Chapter 8) that the linear hull of $V^{\vee k}$ is a dense subspace of $\mathcal{H}^{\vee k}$. Hence, $V^{\vee k}$ separates points of $\mathcal{H}^{\vee k}$. The statements (\ref{PD1}) and (\ref{PD2}) then follow from Theorem \ref{TB}. Moreover, in the fermionic case by Proposition \ref{PC}.(\ref{PC1}) we have $\mu_-=\sqrt{k!}$ and be applying Theorem \ref{TA}.\ref{TA12} and \ref{TA}.\ref{TA22} with $\lambda=\sqrt{k!}$ we obtain
$$\frac{1}{k!}\cdot q_{K^{\otimes k}}(\varphi_-)\leq q_{K^{\wedge k}}(\varphi_-)\leq\frac{1}{k!}q_{K^{\otimes k}}(\varphi_-)\,.$$ 
\end{proof}
\noindent Let us consider some simple but instructive examples for finite dimensional situations; that is for $n:=\dim\,\mathcal{H}<\infty$ and therefore $\mathcal{B}(\mathcal{H})\cong M_n(\C)$. The quantification of entanglement of distinguishable particles in terms of the Minkowski functional $q_{K^{\otimes k}}$ then proceeds as follows.\\
\indent With respect to the trace norm $\norm{\cdot}_1$ on $M_n(\C)$ we can construct the \emph{greatest cross norm} or \emph{projective tensor norm} $\norm{\cdot}_{\pi}$ on the space $\left(M_n(\C)\right)^{\otimes k}$ (compare for example \cite{Ryan}). As was demonstrated by Arveson (compare Theorem 9.1 in \cite{Arveson}) if $\varphi$ is a (normal) state on $\left(M_n(\C)\right)^{\otimes k}$ with density operator $\rho\in\left(M_n(\C)\right)^{\otimes k}$ then $q_{K^{\otimes k}}(\varphi)=\norm{\rho}_{\pi}$.
\begin{example}\label{E1}
\end{example}
\noindent Consider the Hilbert space $\mathbbmss{C}^2\otimes\mathbbmss{C}^2$ and the vector state $\omega_{\xi}$ on $M_2(\mathbbmss{C})\otimes M_2(\C)$ where
$$\xi:=\frac{1}{\sqrt{2}}(e_1\otimes e_2-e_2\otimes e_1)=e_1\wedge e_2$$
for an orthonormal basis $\{e_1,e_2\}$ of $\mathbbmss{C}^2$ (a so called \emph{singlet state}). Because of $\dim\,\C^2\wedge\C^2=1$ and since $\xi\neq 0$ it is clear that $\C^2\wedge\C^2$ is generated by $\xi$. In particular, there are no fermionic entangled states on $\mathcal{B}(\C^2\wedge\C^2)$. Indeed, using the known fact that $q_{K^{\otimes 2}}(\omega_{\xi})=2$ (compare Theorem 14.1 in \cite{Arveson}) and Proposition \ref{PD} we have
$$q_{K^{\wedge 2}}(\omega_{\xi})=\frac{1}{2!}\cdot q_{K^{\otimes 2}}(\omega_{\xi})=1\,.$$
so that $\omega_{\xi}$ is indeed fermionic separable.
\begin{example}
\end{example}
\noindent We have seen in the above example that for dimensional reasons there are no fermionic entangled states on $\mathcal{B}(\C^2\wedge\C^2)$. However, using elementary algebra one can prove that every vector $\xi\in\C^3\wedge\C^3$ is of the form $\eta\wedge\zeta$ for appropriate vectors $\eta,\zeta\in\C^3$ as well. This means that there are no fermionic entangled states on $\mathcal{B}(\C^3\wedge\C^3)$ too. It may be instructive to reproduce this fact by using the methods we developed so far.\\
\indent To this end let $\xi\in\mathbbmss{C}^3\wedge\mathbbmss{C}^3$ be a unit vector and let $\{e_1,e_2,e_3\}$ be an orthonormal basis of $\mathbbmss{C}^3$. Then there are numbers $a,b,c\in\mathbbmss{C}$ such that
\begin{eqnarray*}
\xi&=&a\cdot e_1\wedge e_2+b\cdot e_1\wedge e_3+c\cdot e_2\wedge e_3\\
&=&e_1\otimes\frac{1}{\sqrt{2}}(ae_2+be_3)+e_2\otimes\frac{1}{\sqrt{2}}(-ae_1+ce_3)+e_3\otimes\frac{1}{\sqrt{2}}(-be_1-ce_2)\,.
\end{eqnarray*}
Defining the matrix
$$A:=\left(\begin{array}{ccc}0&a&b\\-a&0&c\\-b&-c&0\end{array}\right)$$
we obtain
$$A^{\ast}A=\mathbbmss{1}_3-\left(\begin{array}{c}c\\-b\\a\end{array}\right)\left(\overline{c},-\overline{b},\overline{a}\right)\,.$$
Since $\left({c},-{b},{a}\right)^T$ is a unit vector the matrix $A^{\ast}A$ is the orthogonal projection of rank 2 onto the orthogonal complement of $\left({c},-{b},{a}\right)^T$. Therefore, the theorems 7.2 and 8.2 in \cite{Arveson} and Proposition 7 in \cite{Rudolph} imply that
\begin{eqnarray*}
q_{K^{\otimes 2}}(\omega_{\xi})=\norm{t_{\xi,\xi}}_{\pi}={\norm{\frac{1}{\sqrt{2}}A}_1}^2=\frac{1}{{2}}\left(\mbox{tr}\sqrt{A^{\ast}A}\right)^2=\frac{2^2}{2}=2\,.
\end{eqnarray*}
Like in the above example we obtain
$$q_{K^{\wedge 2}}(\omega_{\xi})=\frac{1}{2}\cdot q_{K^{\otimes 2}}(\omega_{\xi})=1\,.$$
Thus, $\omega_{\xi}$ is fermionic separable, i.e. there are vectors $\eta,\zeta\in\C^3$ with $\xi=\eta\wedge\zeta$.
\begin{example}\label{E2}
\end{example}
\noindent Let $\mathcal{H}=\mathbbmss{C}^n$ and $\varphi_{\mbox{\footnotesize{tr}}}$ be the tracial (or chaotic) state on $\mathcal{B}\left(\left(\C^n\right)^{\wedge k}\right)$. In particular $k\geq n$. For an arbitrary orthonormal basis $\{e_1,...,e_n\}$ of $\mathbbmss{C}^n$ we have
$$\varphi_{\mbox{\footnotesize{tr}}}=\left(\begin{array}{c}n\\k\end{array}\right)^{-1}\sum_{1\leq i_1\leq...\leq i_k\leq n}{\omega_{e_{i_1}\wedge...\wedge e_{i_k}}}\in\mbox{co}\left\{\omega_{\xi}:\,\xi\in V^{\wedge k}\right\}\subset K^{\wedge k}\,.$$
Hence, $\varphi_{\mbox{\footnotesize{tr}}}$ is fermionic separable and $q_{K^{\otimes k}}(\varphi_{\mbox{\footnotesize{tr}}})$ can be computed according to
$$q_{K^{\otimes k}}(\varphi_{\mbox{\footnotesize{tr}}})=k!\cdot q_{K^{\wedge k}}(\varphi_{\mbox{\footnotesize{tr}}})=k!$$
by Proposition \ref{PD}. This result has been obtained earlier by Maassen \cite{Maassen} using different methods.

\section{Entanglement related to Schmidt and Slater numbers}\label{SC}
In addition to the concepts discussed in the previous section several other notions of entanglement are in use which aim at a somewhat finer entanglement classification. We concentrate on fermionic systems since our techniques work most efficiently in this case. Nevertheless, several aspects can partly be transferred to bosonic particles.\\
\indent In this section we will assume that $\mathcal{H}$ is infinite dimensional in order to avoid unnecessary complications when discussing the various ranks of vectors $\xi\in\mathcal{H}$ below. This is not really a restriction since every finite dimensional Hilbert space can be considered as a subspace of $\mathcal{H}$.
\begin{definition}\label{D3}
Let $\xi\in\mathcal{H}^{\otimes k}$, $\xi_-\in\mathcal{H}^{\wedge k}$ be unit vectors and $\varphi\in\mathcal{B}\left(\mathcal{H}^{\otimes k}\right)_{\ast}$, $\varphi_-\in\mathcal{B}\left(\mathcal{H}^{\wedge k}\right)_{\ast}$ normal states. 
\begin{enumerate}
\item\label{D31}
If $\xi$ can be expressed as a linear combination vectors of $V^{\otimes k}$ we define the \emph{Schmidt rank} of $\xi$ by
$$\mathrm{SR}(\xi):=\min\left\{l\in\N:\,\xi=\sum_{i=1}^l{\lambda_i\eta_i}:\,\eta_i\in V^{\otimes k},\,\lambda_i\in\C\right\}\,.$$
Otherwise, we put $\mathrm{SR}(\xi):=\infty$. The set of unit vectors having Schmidt rank smaller or equal to a number $l\in\N$ is denoted by $V^{\otimes k}_{l}$.
\item
If $\varphi\in\bigcup_{l\in\N}\overline{\mathrm{co}}^{\norm{\cdot}}\set{\omega_{\xi}:\,\xi\in V^{\otimes k}_l}$ the \emph{Schmidt number} of $\varphi$ is defined by
$$\mathrm{SN}(\varphi):=\min\left\{l\in\N:\,\varphi\in\overline{\mathrm{co}}^{\norm{\cdot}}\set{\omega_{\xi}:\,\xi\in V^{\otimes k}_l}\right\}$$
and $\mathrm{SN}(\varphi):=\infty$ if this is not the case.
\item
Analogously, the \emph{Slater rank} $\mathrm{SLR}(\xi_-)$ and the \emph{Slater number} $\mathrm{SLN}(\varphi_-)$ are defined by replacing the set $V^{\otimes k}$ in (i) and (ii) by $V^{\wedge k}$.
\end{enumerate}
\end{definition}
\noindent Note that by the definitions \ref{D1}, \ref{D2}.(\ref{D22}) and \ref{D3}.(\ref{D31}) we have $V^{\otimes k}=V_1^{\otimes k}$ and $V^{\wedge k}=V_1^{\wedge k}$.\\
\indent Schmidt and Slater numbers of normal states in the various contexts are rough measure of how entangled a given normal state is (compare for example \cite{Eckert, Johnston, Schliemann, Terhal}). For bipartite finite dimensional quantum systems Johnston \cite{Johnston} has constructed a cross norm based measure for the detection of states having Schmidt number smaller or equal to a number $l\in\N$ by using the concept of $l$-block positivity. A comparison of this construction to Proposition \ref{PA}.(\ref{PA2}) and (\ref{PA3}) reveals that this is a special case of the following result.
\begin{proposition}\label{PF}
For $l\in\N$ let $K^{\otimes k}_l$ and $K^{\wedge k}_l$ denote the norm-closed convex hulls of the sets
$$\left\{\omega_{\xi,\eta}:\,\xi,\eta\in V^{\otimes k}_l\right\}\quad\mbox{and}\quad\left\{\omega_{\xi,\eta}:\,\xi,\eta\in V^{\wedge k}_l\right\}$$
respectively.
\begin{enumerate}
\item
For a normal state $\varphi\in\mathcal{B}(\mathcal{H}^{\otimes k})_{\ast}$ we have $q_{K^{\otimes k}_l}(\varphi)\geq 1$ and $q_{K^{\otimes k}_l}(\varphi)=1$ if and only if $\mathrm{SN}(\varphi)\leq l$.
\item
For a normal state $\varphi_-\in\mathcal{B}(\mathcal{H}^{\wedge k})_{\ast}$ we have $q_{K^{\wedge k}_l}(\varphi_-)\geq 1$ and $q_{K^{\wedge k}_l}(\varphi_-)=1$ if and only if $\mathrm{SLN}(\varphi_-)\leq l$.
\end{enumerate}
\end{proposition}
\begin{proof}
Both sets $V^{\otimes k}_l$ and $V^{\wedge k}_l$ are balanced and separate the points of $\mathcal{H}^{\otimes k}$ and $\mathcal{H}^{\wedge k}$ respectively since this is the case for $V^{\otimes k}$ and $V^{\wedge k}$. The claim then follows from Theorem \ref{TB}.
\end{proof}
\noindent We have seen in Proposition \ref{PD} that the measures for entanglement of distinguishable and fermionic particles are related to each other by a factor of $k!$.  In the sequel we generalize this result to the measures $q_{K_l^{\otimes k}}$ and $q_{K_l^{\wedge k}}$ for $l\geq 2$ and $k=2$. The restriction to the bipartite case is necessary since the proofs strongly rely on the following canonical forms for state vectors.
\begin{proposition}
Let $\mathcal{H}$ be a Hilbert space and $\xi\in\mathcal{H}\otimes\mathcal{H}$. Then there are a sequence of numbers $\lambda_n\geq 0$ with $(\lambda_n)_{n\in \N}\in\ell^2(\N)$ and orthonormal systems $(e_n)_{n\in\N}, (f_n)_{n\in\N}\subset\mathcal{H}$ such that the following holds:
\begin{enumerate}
\item
The vector $\xi$ can be written in the form
$$\xi=\sum_{i=1}^{\infty}\lambda_n e_n\otimes f_n\,.$$
This representation is called a \emph{Schmidt decomposition} of $\xi$ and the numbers $\lambda_n$ are called \emph{Schmidt coefficients}.
\item
If $\xi\in\mathcal{H}\wedge\mathcal{H}$ the numbers $\lambda_n$ and the vectors $e_n,f_n$ can be chosen such that the orthonormal systems $(e_n)_{n\in\N}$ and $(f_n)_{n\in\N}$ are orthogonal to each other and
$$\xi=\sum_{i=1}^{\infty}\lambda_n e_n\wedge f_n\,.$$
This representation is called a \emph{Slater decomposition} of $\xi$ with \emph{Slater coefficients} $\lambda_n$.
\end{enumerate}
In both cases the sequence $\lambda_n$ is unique up to ordering. Moreover, the number of non-vanishing coefficients $\lambda_n$ in 1. and 2. equals the Schmidt rank and the Slater rank of $\xi$ respectively.
\end{proposition}
\begin{proof}
The first statement is a simple consequence of the normal form for compact operators on a Hilbert space whereas (ii) has been proven by Schliemann et al. \cite{Schliemann} for finite dimensional Hilbert spaces. However, using a block diagonal argument for blocks of constant Schmidt coefficients the general case can be easily reduced to situations involving only finite dimensional Hilbert spaces:\\
For a fixed Schmidt coefficient $\lambda\neq 0$ let
$$N_{\lambda}:=\{n\in\mathbbmss{N}:\,\lambda_n=\lambda\}\,.$$
Since the sequence of Schmidt coefficients converges to zero the set $N_{\lambda}$ is finite and the space
$$\mathcal{K}_{\lambda}:=\mbox{lin}\{e_n:\,n\in N_{\lambda}\}$$
has finite dimension. Moreover, since $\xi\in\mathcal{H}\wedge\mathcal{H}\subset\mathcal{H}\otimes\mathcal{H}$ we have
$$\sum_{n=1}^{\infty}\lambda_ne_n\otimes f_n=\xi=\sum_{n=1}^{\infty}{\lambda_n(-f_n)\otimes e_n}$$
and using the fact that both sides of the above equation constitute a Schmidt decomposition of $\xi$ with the same Schmidt coefficients it follows
$$\mathcal{K}_{\lambda}=\mbox{lin}\{f_n:\,n\in N_{\lambda}\}\,.$$
Furthermore, since $\mathcal{K}_{\lambda}\bot \mathcal{K}_{\mu}$ for every other Schmidt coefficient $\mu\neq \lambda$ we have
$$\sum_{n\in N_{\lambda}}\lambda e_n\otimes f_n=-\sum_{n\in N_{\lambda}}\lambda f_n\otimes e_n\,.$$
This shows that
$$\sum_{n\in N_{\lambda}}\lambda e_n\otimes f_n\in\mathcal{K}_{\lambda}\wedge\mathcal{K}_{\lambda}$$
and the claim now follows from the finite dimensional case.
\end{proof}
\noindent We have the following analogue of Proposition \ref{PC} concerning the geometrical properties of the sets $V^{\otimes 2}_l$ and $V^{\wedge 2}_l$.
\begin{proposition}\label{PE}
For $l\in\N$, $l \geq 2$ let us define
$$\mu_{-,l}:=\sup\left\{r\geq 1:\,rP_{-}V^{\otimes 2}_l\subset B_{1,\mathcal{H}^{\otimes 2}}\right\}\,.$$
\begin{enumerate}
\item\label{PE1}
We have $\mu_{-,l}=1$ and
$$V^{\wedge 2}_l=\left(\sqrt{2}P_-V^{\otimes 2}_l\right)\cap S_{1,\mathcal{H}^{\otimes 2}}\,.$$
\item
The projection $P_-$ is compatible with $\left(V_l^{\otimes 2},\sqrt{2}\right)$.
\item
The mapping
$$\mathcal{B}(\mathcal{H}^{\otimes 2})\ni x\mapsto P_{\pm}xP_{\pm}\in\mathcal{B}(\mathcal{H}^{\otimes 2})$$
is a $\norm{\cdot}_{K^{\otimes 2}_{l}}$-contraction.
\end{enumerate}
\end{proposition}
\begin{proof}$\,$
\begin{enumerate}
\item 
For $l\geq 2$ the projection $P_-$ has fixed-points in $V_l^{\otimes 2}$, for example vectors of the type
$$\frac{1}{\sqrt{2}}(e_1\otimes e_2-e_2\otimes e_1)$$
for an orthonormal system $\{e_1,e_2\}\subset\mathcal{H}$.\\
\indent Clearly, projecting a vector $\xi\in V_l^{\otimes 2}$ onto $\mathcal{H}^{\wedge 2}$ by $P_-$ results in a vector having Slater rank less than or equal to $l$; that is $\left(\sqrt{2}P_-V_l^{\otimes 2}\right)\cap S_{1,\mathcal{H}^{\otimes 2}}\subset V_l^{\wedge 2}$. On the other hand, every vector $\xi\in V_l^{\wedge 2}$ with Slater decomposition $\xi=\sum_{i=1}^{l}{\lambda_i e_i\wedge f_i}$ can be written in the form
$$\xi=\sum_{i=1}^{l}{\frac{\lambda_i}{\sqrt{2}}(e_i\otimes f_i-f_i\otimes e_i)}=\sqrt{2}P_-\sum_{i=1}^{l}{\lambda_ie_i\otimes f_i}\in \sqrt{2}P_-V_l^{\otimes 2}\,.$$
\item
As mentioned in (i) we have $\mbox{SLR}(P_-\xi)\leq\mbox{SR}(\xi)$ for all $\xi\in\mathcal{H}\otimes\mathcal{H}$. Therefore, it is clear that
$$\frac{P_-\xi}{\norm{P_-\xi}}\in V_l^{\wedge 2}\quad\mbox{for all }\xi\in V^{\otimes 2}_l\mbox{ with }P_-\xi\neq 0$$
so that $P_-$ is compatible with $\left(V^{\otimes 2}_l,\sqrt{2}\right)$.
\item
Using that for every permutation $\pi\in S_k$ the corresponding permutation operator induces a bijection from $V^{\otimes 2}_l$ onto itself the proof of Proposition \ref{PC}.(\ref{PC3}) can be mimicked literally.
\end{enumerate}
\end{proof}
\noindent Taking a look at Proposition \ref{PE}.(\ref{PE1}) we see that the constant $\mu_{-,l}$ no longer equals the corresponding constant $\mu_-=\sqrt{2!}=\sqrt{2}$ from Proposition \ref{PC}.(\ref{PC1}). As a consequence, for $l\geq 2$ the Minkowski functionals $q_{K^{\otimes 2}_l}$ and $q_{K^{\wedge 2}_{l}}$ are no longer multiples of each other as it is the case for $l=1$ according to Proposition \ref{PD}. Nevertheless, they are still ''equivalent'' in the following sense. 
\begin{proposition}\label{PG}
For $l\in\N$, $l\geq 2$ and a normal state $\varphi_-\in\mathcal{B}(\mathcal{H}^{\wedge 2})_{\ast}$ we have the chain of inequalities
$$\frac{1}{2}\cdot q_{K^{\otimes 2}_l}(\varphi_-)\leq q_{K^{\wedge 2}_l}(\varphi_-)\leq q_{K^{\otimes 2}_l}(\varphi_-)$$
and on both sides equality can be achieved for an appropriate choice of $\varphi_-$. In particular, these inequalities can in general not be sharpened any further.
\end{proposition}
\begin{proof}
Both inequalities follow from Proposition \ref{PE} and Theorem \ref{TA}.\\
In order to achieve equality we choose mutually orthogonal orthonormal systems $\{e_1,...,e_l\}, \{f_1,...,f_l\}\subset\mathcal{H}$ and consider the family of vectors
$$\xi_k:=\frac{1}{\sqrt{2k}}\sum_{i=1}^k{(e_i\otimes f_i-f_i\otimes e_i)}=\frac{1}{\sqrt{k}}\sum_{i=1}^k{e_i\wedge f_i}\in\mathcal{H}\wedge\mathcal{H}\quad\mbox{for }1\leq k\leq l\,.$$
Then we have $\mbox{SR}(\xi_k)=2k$ and $\mbox{SLR}(\xi_k)=k$. Hence the corresponding normal states $\omega_{\xi_k}$ satisfy $\mbox{SN}(\omega_{\xi_k})\leq 2k$ and $\mbox{SLN}(\omega_{\xi_k})\leq k$. Therefore, equality on the right-hand-side can be achieved as follows.\\
\indent For $1\leq k \leq\left\lfloor\frac{l}{2}\right\rfloor$ (note that this condition can always be satisfied by $k$ since $l\geq 2$) we have $\mbox{SN}(\omega_{\xi_k})\leq l$ and $\mbox{SLN}(\omega_{\xi_k})\leq l$. It follows
$$q_{K_l^{\otimes 2}}(\omega_{\xi_k})=1=q_{K_l^{\wedge 2}}(\omega_{\xi_k})$$
by Proposition \ref{PF}.\\
\indent In order to achieve equality on the left-hand-side we consider the vector $\sqrt{2}\xi_l$ (with Schmidt rank $2l$). As a first step we show that the corresponding rank-one-operator $t_{\sqrt{2}\xi_l}$ satisfies
$$\norm{t_{\sqrt{2}\xi_l}}_{K^{\otimes 2}_l}\leq 1\,.$$
Indeed, choosing an arbitrary unit vector $\eta\in V^{\otimes 2}_l$ with Schmidt decomposition
$$\eta=\sum_{j=1}^l{\lambda_je_j'\otimes f_j'}$$
a two-fold application of the Cauchy-Schwarz inequality followed by Bessel's inequality with respect to the orthonormal system $\set{e_1,...,e_l,f_1,...,f_l}$ yields
\begin{eqnarray*}
\left|\left\langle\sqrt{2}\xi_l,\eta\right\rangle\right|^2&=&\left|\left\langle\frac{1}{\sqrt{l}}\sum_{i=1}^l{(e_i\otimes f_i-f_i\otimes e_i)},\sum_{j=1}^l{\lambda_je_j'\otimes f_j'}\right\rangle\right|^2\\
&=&\frac{1}{l}\left|\sum_{i=1}^l\sum_{j=1}^l\lambda_j\langle e_i,e_j'\rangle\langle f_i,f_j'\rangle-\lambda_j\langle f_i,e_j'\rangle\langle e_i,f_j'\rangle\right|^2\\
&=&\frac{1}{l}\left|\sum_{j=1}^l\lambda_j\cdot\left\langle\sum_{i=1}^l\langle e_i,e_j'\rangle f_i-\langle f_i,e_j'\rangle e_i,f_j'\right\rangle\right|^2\\
&\leq&\frac{1}{l}\underbrace{\left(\sum_{j=1}^l\lambda_j^2\right)}_{=1}\left(\sum_{j=1}^l\left|\left\langle\sum_{i=1}^l\langle e_i,e_j'\rangle f_i-\langle f_i,e_j'\rangle e_i,f_j'\right\rangle\right|^2\right)\\
&=&\frac{1}{l}\sum_{j=1}^l\left|\left\langle\sum_{i=1}^l\langle e_i,e_j'\rangle f_i-\langle f_i,e_j'\rangle e_i,f_j'\right\rangle\right|^2\\
&\leq&\frac{1}{l}\sum_{j=1}^l\norm{\sum_{i=1}^l\langle e_i,e_j'\rangle f_i-\langle f_i,e_j'\rangle e_i}^2\cdot\norm{f_j'}^2\\
&=&\frac{1}{l}\sum_{j=1}^l\sum_{i=1}^l\left(|\langle e_i,e_j'\rangle|^2+|\langle f_i,e_j'\rangle|^2\right)\\
&\leq&\frac{1}{l}\sum_{j=1}^l{\norm{e_j'}^2}=1\,.
\end{eqnarray*}
It follows
\begin{eqnarray*}
\norm{t_{\sqrt{2}\xi_l}}_{K^{\otimes 2}_l}&=&\sup\left\{\left|\left\langle t_{\sqrt{2}\xi_l}\eta,\zeta\right\rangle\right|:\,\eta,\zeta\in V^{\otimes 2}_l\right\}\\
&=&\sup_{\eta\in V_l^{\otimes 2}}\sup_{\zeta\in V_l^{\otimes 2}}\left|\left\langle\sqrt{2}\xi_l,\zeta\right\rangle\left\langle\eta,\sqrt{2}\xi_l\right\rangle\right|\\
&=&\sup_{\eta\in V_l^{\otimes 2}}\left|\left\langle\sqrt{2}\xi_l,\eta\right\rangle\right|\cdot\sup_{\zeta\in V_l^{\otimes 2}}\left|\left\langle\sqrt{2}\xi_l,\zeta\right\rangle\right|\\
&=&\sup\left\{\left|\left\langle \sqrt{2}\xi_l,\eta\right\rangle\right|^2:\,\eta\in V^{\otimes 2}_l\right\}\leq 1\\
\end{eqnarray*}
as claimed. By Proposition \ref{PA}.(\ref{PA3}) $q_{K^{\otimes 2}_l}(\omega_{\xi_l})$ can therefore be estimated according to
\begin{eqnarray*}
q_{K^{\otimes 2}_l}(\omega_{\xi_l})&=&\sup\left\{|\omega_{\xi_l}(x)|:\,x\in\mathcal{B}(\mathcal{H}),\,\norm{x}_{K^{\otimes 2}_l}\leq 1\right\}\\
&\geq&\left|\omega_{\xi_l}\left(t_{\sqrt{2}\xi_l}\right)\right|=2|\langle t_{\xi_l}\xi_l,\xi_l\rangle|=2\,.
\end{eqnarray*}
Since $\mbox{SLN}(\omega_{\xi_l})\leq l$, it follows
$$1\leq\frac{1}{2}\cdot q_{K_l^{\otimes}}(\omega_{\xi_l})\leq q_{K_l^{\wedge 2}}(\omega_{\xi_l})=1\,.$$
\end{proof}
\noindent According to Proposition \ref{PG} for $l\geq 2$ the measures $q_{K_l^{\otimes 2}}$ and $q_{K_l^{\wedge 2}}$ are not as strongly coupled as it is the case for $l=1$ (compare Proposition \ref{PD}). Nevertheless, Proposition \ref{PG} can be interpreted as an equivalence of the generalized norms $q_{K_l^{\otimes 2}}$ and $q_{K_l^{\wedge 2}}$. In particular, $q_{K_l^{\otimes 2}}$ is finite if and only if $q_{K_l^{\wedge 2}}$ is finite. 

\section{Summary and Conclusions}
In this paper we have demonstrated how the detection and quantification of many different kinds of quantum entanglement can be dealt with by a generalization of Arveson's unifying approach \cite{Arveson}. We have constructed new and general entanglement measures for fermionic and bosonic particles which may be interpreted as generalized norms and work for mixed states and multipartite systems regardless of the dimensions. These entanglement measures also constitute necessary and sufficient computational separability criteria.\\
\indent We have seen that the measures for distinguishable and fermionic particles are multiples of each other by a factor of $k!$ where $k$ is the number of particles. When it comes to the quantification of fermionic entanglement this fact allows to treat the particles as if they were distinguishable. Moreover, the amount of entanglement for distinguishable particles emerging from fermi statistics alone is at least $k!$.\\\newpage
\indent Furthermore, we have introduced analog entanglement measures for more sophisticated notions of entanglement related to Schmidt and Slater numbers. Thereby, a result of Johnston \cite{Johnston} on the detection and quantification of Schmidt number entanglement which was formulated for bipartite finite dimensional systems has been generalized to arbitrary quantum systems. In the bipartite case we have demonstrated that the measures for Schmidt and Slater number entanglement satisfy a chain of inequalities which may be interpreted as an equivalence of generalized norms. Furthermore, it has been shown that these inequalities can in general not be sharpened any further.\\
\indent It is an important open question for further research whether the relations according to Proposition \ref{PG} are still valid for multipartite quantum systems where the possibility of performing Schmidt and Slater decompositions is no longer given. On the other hand, using the Schmidt decomposition it is possible to find explicit formulas for the amount of entanglement of pure bipartite states for distinguishable particles \cite{Rudolph2}. According to the proof of Proposition \ref{PG} it seems likely that this is possible for the measures $q_{K^{\otimes 2}_l}$ and $q_{K^{\wedge 2}_l}$ as well. However, many aspects concerning the subtle geometric properties of these measures are not yet fully understood. This will be another subject of further studies as well.

\section*{Acknowledgments}
We are grateful to Hans Maassen for numerous stimulating and fruitful discussions especially concerning fermionic entanglement.

\end{document}